\documentclass{siamltex}

\usepackage{packages}
\usepackage{array}

\title{The Kirkwood closure point process: A solution of the Kirkwood-Salsburg equations for negative activities\thanks{The research leading to this work
    has been done within the 
    Collaborative Research Center TRR~146; corresponding funding 
    by the DFG is gratefully acknowledged.}}
\author{Fabio Frommer\thanks{Institut f\"ur Mathematik, Johannes
    Gutenberg-Universit\"at Mainz, 55099 Mainz, Germany
    ({\tt fabiofrommer@uni-mainz.de})}}
\date{\today}

\makeatletter
\renewcommand\@biblabel[1]{#1.}
\def\proofthm[#1]{\par{\it Proof of Theorem #1}. \ignorespaces}

\newcommand{\nocontentsline}[3]{}
\let\origcontentsline\addcontentsline
\newcommand\stoptoc{\let\addcontentsline\nocontentsline}
\newcommand\resumetoc{\let\addcontentsline\origcontentsline}
\begin{document}

\maketitle

\stoptoc
\begin{abstract} 
The Kirkwood superposition is a well-known tool in statistical physics to approximate the $n$-point correlation functions for $n\geq 3$ in terms of the density $\rho $ and the radial distribution function $g$ of the underlying system. However, it is unclear whether these approximations are themselves the correlation functions of some point process. If they are, this process is called the Kirkwood closure process. For the case that $g$ is the negative exponential of some nonnegative and regular pair potential $u$ existence of the Kirkwood closure process was proved by Ambartzumian and Sukiasian. This result was generalized to the case that $u$ is a locally stable and regular pair potential by Kuna, Lebowitz and Speer, provided that $\rho$ is sufficiently small. In this work, it is shown that it suffices for $u$ to be stable and regular to ensure the existence of the Kirkwood closure process. Furthermore, for locally stable $u$ it is proved that the Kirkwood closure process is Gibbs and that the kernel of the GNZ-equation satisfies a Kirkwood-Salsburg type equation.    
\end{abstract}


\begin{AMS}
{\sc 82B21, 60G55}
\end{AMS}

\section{Introduction}
In classical statistical physics, point processes are often used to describe the distribution of interacting particles in equilibrium.
Often, so-called \emph{Gibbs measures} are used. In these models the energy of a configuration of particles is calculated via some interaction potentials and a configuration is more likely to be observed when the associated energy is low. However, in general, it is not possible to measure these interaction potentials nor calculate them easily from given snapshots of these configurations, see e.g.~\cite{HansenMcDonald13}.\\
In practice, the available data are the so-called \emph{$n$-point correlation functions} $\rho^{(n)}$ of the underlying point process. However, while it is possible to calculate them for arbitrary $n$, these calculations get very computationally expensive as soon as $n>2$ as good statistics require long simulation times and $n$-tuples of particles have to be counted. Thus, commonly the \emph{Kirkwood superposition approximation}, introduced by Kirkwood in \cite{Kirkwood42}, is used, cf.~\cite{HansenMcDonald13}, to approximate the higher-order correlation functions, i.e.
\begin{align}\label{eq:approxgen}
	\rho^{(n)}(\xx_n) \approx
    \rho^n \prod_{1\leq i<j\leq n} g(x_i-x_j).
\end{align}
Here $\xx_n=(x_1,\dots,x_n)$ and $g=\rho^{(2)}/\rho^2$ is the so-called \emph{radial distribution function} of the point process. In \cite{Ambartzumian91} the question has been raised whether there is a point process $\KK$ whose correlation functions are given by the right-hand side of \req{approxgen}. This means the closed form expression of the  correlation functions of the process $\KK$ are given by the Kirkwood superposition, 
thus this point process $\KK$ is called the \emph{Kirkwood closure process}. In this work sufficient conditions for the existence of $\KK$ are investigated.\\
This question is related to an interesting inverse problem, namely, a realizability problem for point processes, see \cite{Kuna07}:\\
\glqq Given $\rho>0$ and a nonnegative function $g$, does there exist a point process with density $\rho$ and radial distribution function $g$?\grqq{} \\
The Kirkwood closure process is one possible ansatz for the solution of this problem.\\
For the case that $g\leq 1$ Ambartzumian and Sukiasian showed in \cite{Ambartzumian91} that the Kirkwood closure process exists when $\rho$ is small enough. Later, using a different technique this result was extended by Kuna, Lebowitz and Speer in \cite{Kuna07}.\\
In the language of statistical mechanics Ambartzumian and Sukiasian showed the existence of the Kirkwood closure when $g=e^{-u}$ where $u$ is some nonnegative and \emph{regular} pair potential and Kuna, Lebowitz and Speer extended the result for the case that $u$ is a pair potential which is \emph{locally stable} (e.g. when $u$ has a \emph{hard-core}, i.e.~$u=+\infty$ around the origin) and regular. In this work a connection between the well-known \emph{Kirkwood-Salsburg equations} and the Kirkwood closure process is used to show existence of the latter when $u$ is a \emph{stable} and {regular} pair potential. In fact, the so-called \emph{Janossy densities} of the Kirkwood closure process are (up to a factor) the solutions of the Kirkwood-Salsburg equations for a negative activity. In particular, this solution has many well-known properties, cf.~\cite{Ruelle69}.\\
The outline is as follows: After introducing the setting in Section \ref{sec:setting}, the existence of the Kirkwood closure is proved in Section \ref{sec:results}. In Section \ref{sec:GNZresult} the Gibbsianness of the Kirkwood closure is discussed and in the last Section generalizations of to higher order closures are discussed.
\section{Setting}\label{sec:setting}
\subsection{The Kirkwood closure process}
Any probability measure $\PP$ on the space of configurations
\begin{align*}
    \Gamma \,=\, \bigl\{\, \gamma \subset \mathbb{R}^d \ \Big\vert\ 
                   \Delta  \subset \R^d \text{ bounded}
                   \,\Rightarrow \, N_\Delta(\gamma)<+\infty \bigr\}\,,
\end{align*}
equipped with the $\sigma$-algebra $\mathcal{F}:=\sigma (N_\Delta \mid \Delta \subset \R^d \text{ bounded})$ is called a \emph{point process}. Here $N_\Delta(\gamma)= \#(\gamma_\Delta)$ ($\gamma_\Delta=\gamma\cap \Delta$) is the number of elements of $\gamma$ in $\Delta$. 
$\Gamma_0=\{\gamma\in\Gamma\mid \#\gamma<+\infty\} $ denotes the space of finite configurations. The elements of a family of symmetric functions $(j^{(n)}_\Lambda)_{n\geq 0, \,\Lambda \subset \R^d \text{ bounded}}$ are called the \emph{Janossy densities} of $\PP$, if for every $F\colon \Gamma \to [0,+\infty)$ such that $F(\gamma)=F(\gamma_\Lambda)$ for every $\gamma\in\Gamma$ there holds
\begin{align}
	\int_{\Gamma} F(\gamma) \dP (\gamma)
	= \sum_{n=0}^\infty\frac{1}{n!}
	\int_{\Lambda^n}F(\{\xx_n\})\jan{n}(\xx_n)\dxx_n
\end{align}
where the term for $n=0$ is understood to be $ F(\emptyset)\jan{0}$. Any function with the property $F(\gamma)=F(\gamma_\Lambda)$ for some bounded $\Lambda \subset \R^d$ is called \emph{local}.
If the Janossy densities of a point process exist, they are unique up to (Lebesgue) null-sets and determine $\PP$ completely.\\ 
The elements of a family of symmetric functions  $(\rho^{(n)})_{n\in\N}$ are called the \emph{correlation functions} of $\PP$, if for every $n\in\N$ and $F\colon (\R^d)^n\to [0,+\infty)$ there holds
\begin{align}
	\int_{\Gamma} 
	\sum_{x_1,\dots,x_n\in \gamma \atop x_i\neq x_j}
	 F(\xx_n) \dP (\gamma)
	 =
	 \int_{(\R^d)^n}F(\xx_n) \rho^{(n)}(\xx_n)\dxx_n.
\end{align}
Also note the formula
\begin{align}\label{eq:janossytocorr}
    \rho^{(n)}(\xx_n) = \sum_{k=0}^\infty \frac{1}{k!}\int_{\Lambda^k}j_\Lambda^{(n+k)}(\xx_n,\yy_k)\dyy_k, \qquad \xx_n \in  \Lambda^n.
\end{align}
Here $j_\Lambda^{(n+k)}(\xx_n,\yy_k) = j_\Lambda^{(n+k)}(x_1,\dots,x_n,y_1,\dots,y_k)$ for brevity. If the point process $\PP$ is \emph{stationary}, then the correlation functions are \emph{translationally invariant}, and one can write $\rho^{(n)}= \rho^n g^{(n)}$ for appropriate functions $g^{(n)}$ depending on $n-1$ variables, where $\rho$ is the so-called \emph{intensity} or \emph{density} of the point process. For $n=2$ the function $g=g^{(2)}$ is the so-called \emph{radial distribution function}.\\ 
As mentioned in the introduction, a  point process $\KK_{\varsigma,\phi}$ is called \emph{Kirkwood closure process}, if it has correlation functions and there is a $\varsigma>0$ and an even nonnegative function $\phi\colon\R^d\to [0,+\infty)$ such that 
\begin{align}\label{eq:easycorrelations}
    \rho^{(n)}(\xx_n) =  
    \varsigma^n \prod_{1\leq i < j\leq n} \phi(x_i-x_j),
\end{align}
where the empty product is understood to be equal to one. In particular, this means that for the Kirkwood closure the approximation \req{approxgen} is an equality. The existence of the Kirkwood closure will be discussed in Section \ref{sec:results}.\\
The correlation functions $(\rho^{(n)})_{n\geq 1}$ of a point process $\PP$ satisfy \emph{Ruelle's bound}, if there is a $\xi >0$ such that 
\begin{align}\label{eq:Rbound}
	\rho^{(n)}(\xx_n) \leq \xi^n. \tag{$\mathcal{R}_\xi$}
\end{align}
In this case it is said that $\PP $ satisfies condition \req{Rbound}. Any point process $\PP$ satisfying condition \req{Rbound} has a number of nice properties. Firstly, in this case the correlation functions determine $\PP$ uniquely, see \cite{KunaPhD}. Secondly, $\PP$ is supported on a set of \glqq nice\grqq{} configurations. Namely, any point process $\PP$ satisfying condition \req{Rbound} is supported on the \emph{tempered configurations}
\begin{align}\label{eq:tempconf}
    \Gamma_* :=  \bigcup_{M \geq 1} 
    \bigcap_{n \geq 0}
    \left\lbrace
    \gamma \in \Gamma \mid N_{\Delta_n} (\gamma) \leq M \lebesgue (\Delta_n) 
    \right\rbrace
\end{align}
where $\Delta_n = \{x\in\R^d \mid n\leq |x| < n+1\}$ and $\lebesgue(\cdot) $ denotes the Lebesgue measure on $\R^d$, i.e.~$\PP(\Gamma_*)=1$, cf.~e.g. Theorem 2.5.4 of \cite{KunaPhD}. In this case $\PP$ is called \emph{tempered}. Lastly, for any point process satisfying condition \req{Rbound} the inverse to \req{janossytocorr} holds, i.e.~for any bounded $\Lambda\subset \R^d$ there holds
\begin{align}\label{eq:janossyinverse}
    j_\Lambda^{(n)}(\xx_n) = \sum_{k=0}^\infty\frac{(-1)^k}{k!} \int_{\Lambda^k}\rho^{(n+k)}(\xx_n,\yy_k)\dyy_k, \qquad \xx_n \in  \Lambda^n
\end{align}
where for $n=0$ the term $\rho^{(0)}=1$, see e.g.~\cite{KunaPhD}. In fact, \req{janossyinverse} can also be used to define a point process:
\begin{theoremold}\label{thm:lenard}[Lenard \cite{Lenard75}]
Let $(\rho^{(n)})_{n\geq 1}$ be a family of nonnegative symmetric functions that satisfy \req{Rbound} for some $\xi>0$ such that for all $n\in\N$, all bounded $\Lambda\subset \R^d$ and all $\xx_n\in\Lambda^n$
\begin{align}\label{eq:lenpos1}
    \sum_{k=0}^\infty\frac{(-1)^k}{k!} \int_{\Lambda^k}\rho^{(n+k)}(\xx_n,\yy_k)\dyy_k \geq 0
\end{align}
and 
\begin{align}\label{eq:lenpos2}
    1+\sum_{k=1}^\infty\frac{(-1)^k}{k!} \int_{\Lambda^k}\rho^{(k)}(\yy_k)\dyy_k \geq 0.
\end{align}
Then there exists a point process $\PP$ with correlation functions $(\rho^{(n)})_{n\geq 1}$.
\end{theoremold}\noindent
The conditions \req{lenpos1} and \req{lenpos2} are called \emph{Lenard positivity}. In general, it is not easy to check whether a family $(\rho^{(n)})_{n\geq 1}$ satisfies the Lenard positivity condition. However, for the correlation functions of the Kirkwood closure process sufficient conditions for Lenard positivity have been given. First by Ambartzumian and Sukiasian in \cite{Ambartzumian91} and later these were generalized by Kuna, Lebowitz and Speer in \cite{Kuna07}.
Ambartzumian and Sukiasian relied on an approach using a cluster expansion and Kuna et al.~used an ansatz via modified Kirkwood-Salsburg equations related to the Mayer-Montroll equations, which are both well-known tools from classical statistical mechanics. As previously mentioned, in this work an approach using properties of the {Kirkwood-Salsburg} equations is used to extend their results.
 
\subsection{The Kirkwood-Salsburg operator}\label{ss:KSOP}
The \emph{Kirkwood-Salsburg equations} are a well-known tool for \emph{grand-canonical Gibbs measures}, cf.~\cite{Ruelle69}. Let $u\colon\R^d\to \R\cup\{+\infty\}$ be an even function bounded from below to which a translationally invariant Hamiltonian $H\colon \Gamma_0\to \R\cup\{+\infty\}$ is associated by
\begin{align}\label{eq:Hpair}
    H (\gamma) = \frac{1}{2}\sum_{x\neq y \in \gamma}u(x-y).
\end{align}
The function $u$ is called a translationally invariant \emph{pair potential}. For $\beta>0$ the \emph{Mayer function} of $u$ at \emph{inverse temperature} $\beta$ is defined as
\begin{align}\label{eq:mfct}
    f_\beta\,(x) = e^{-\beta u(x)}-1.
\end{align}
Throughout it is assumed that $u$ is \emph{regular}, i.e.~that 
\begin{align}\label{eq:cbeta}
    \cbeta := \int_{\R^d}\left|\mayer(x)\right|\dx < +\infty
\end{align}
for all $\beta >0$. In fact, if there is a $\beta_0$ such that $C_{\beta_0}(u)<+\infty$, then $\cbeta$ is finite for all $\beta>0$, cf.~\cite{Ruelle69}. It will further be assumed that the pair potential $u$ (and thus the Hamiltonian $H$) is stable, meaning there is a $B>0$ such that
\begin{align}\label{eq:stability}
    H(\gamma) \geq -B \#\gamma.
\end{align}
\begin{remark}
A sufficient condition for $u$ to be stable and regular, is that $u$ is of \emph{Lennard-Jones type}, i.e. that there exist $r_0>0$, $\alpha>d$, and $ C>c>0$ such that 
\begin{align*}
    u(x) \geq  c|x|^{-\alpha}, \quad |x| < r_0,\quad
    \text{ and }\quad
    |u(x)| \leq C|x|^{-\alpha}, \quad |x| \geq  r_0.
\end{align*}
\end{remark}\noindent \\
The \emph{interaction} between $\eta \in\Gamma_0$ and $\gamma\in\Gamma$ is defined by
\begin{align}\label{eq:winteraction}
    W(\eta\mid \gamma) :=
    \begin{dcases}
        \sum_{x\in\eta,y\in\gamma} u(x-y), \qquad \,\,\text{ if }\quad
        \sum_{x\in\eta,y\in\gamma} |u(x-y)| < +\infty \\
        +\infty,\qquad\qquad\qquad\quad \text{ otherwise.}
    \end{dcases}
\end{align}
For two finite configurations $\eta,\gamma\in \Gamma_0$, there holds
\begin{align}\label{eq:intprop}
    H(\eta\cup \gamma)
    =H(\eta)
    +W(\eta\mid\gamma)
    + H(\gamma).
\end{align}
From \req{intprop} it follows that if $H(\eta)=+\infty $ then $H(\eta \cup \{x\})=+\infty$ for all $x\in \R^d$, this means that $H$ is \emph{hereditary}. Since $u$ is assumed to be a stable pair potential, every configuration $\xx_n$ has an element $x_{i_*}$ with $i_*=i_*(\xx_n)\in\{1,\dots,n\}$ such that 
\begin{align}\label{eq:locstab}
    W(\{x_{i_*}\}\mid \{\xx_{n-1}'\})=
    \sum_{\substack{i=1 \\ i\neq i_*}}^n u (x_i-x_{i_*})
    \geq -2B
\end{align}
where $\xx_{n-1}'=(x_1,\dots,x_{i_*-1},x_{i_*+1},\dots,x_n)$ is the configuration of the remaining elements, cf.~\cite{Ruelle69}. In case there is more than one possible choice such that \req{locstab} holds, let $i_*$ be the smallest index with this property. If this property holds for every $n$ and every choice of $i$, i.e.~for any $n\geq 1$ and $x,x_1,\dots,x_n\in\R^d$ there holds
        \begin{align}\label{eq:localstab2}
            W(\{x\}\mid \{\xx_{n}\})=
            \sum_{   i=1}^n u (x_i-x)
            \geq-2B,
        \end{align}
then $u$ is called \emph{locally stable}, cf.~\cite{Kendall99}. Note that local stability is more restrictive than stability as every locally stable pair potential is stable.
\\
For $\zeta>0$, let
\begin{align}\label{eq:Exi}
    E_{\zeta} :=\left\lbrace 
    \bo= (\omega^{(n)})_{n\geq 1}
    \;\middle|\;
    \omega^{(n)}\colon (\R^d)^n\to \C, \,\,\,
    \|\bo\|_\zeta<+\infty
    \right\rbrace
\end{align}
be the Banach space of sequences of complex $L^\infty$-functions with an increasing number of variables, for which the norm
\begin{align*}
    \|\bo\|_\zeta:=
    \sup_{n\geq 1}\left(\zeta^n\|\omega^{(n)}\|_\infty\right)
\end{align*}
is finite and introduce the \emph{Kirkwood-Salsburg} operator $\K \colon E_{\cbeta} \to E_{e^{-2\beta B}\cbeta}$ as
\begin{align}\label{eq:Kdef1}
    (\K\bo)^{(1)}(x )=
    \sum_{k=1}^\infty\frac{1}{k!}\int_{(\R^d)^k}
    \prod_{i=1}^k \mayer(x-y_i)
    \omega^{(k)}(\yy_k)\dyy_k
\end{align}
and for $n\geq 1$ as
\begin{align}\label{eq:Kdef2}
     (\K\bo)^{(n+1)}(x,\xx_{n})=  
     e^{-\beta W(\{x \}\mid \{\xx_{n}\})}
    \left(
    \omega^{(n)}(\xx_{n})+
    \sum_{k=1}^\infty\frac{1}{k!}\int_{(\R^d)^k}
    \prod_{j=1}^k \mayer(x-y_j)
    \omega^{(n+k)}(\xx_{n},\yy_k)\dyy_k\right).
\end{align}
Defining the permutation operator $\BPi \colon \Ec \to \Ec$ by $(\BPi\bo)^{(n)}(\xx_n)=\omega^{(n)}(x_{i_*},\xx_{n-1}') $, one finds by \req{locstab} that
\begin{align*}
    \sup_{n\geq 1} \cbeta^n \|(\BPi\K\bo)^{(n)}\|_\infty 
    &\leq
    \cbeta^n e^{2\beta B} \sum_{k=0}^\infty \frac{1}{k!}
    \int_{(\R^d)^k}
    \prod_{j=1}^k |\mayer(x_{i_*}-y_j)|
    \cbeta^{-n-k+1}
    \|\bo\|_{\cbeta}
    \dyy_k \\
    &\leq e^{2\beta B+1}\cbeta 
\end{align*}
and thus $\BPi\K \colon \Ec\to \Ec $ is well-defined with $\|\BPi\K\|_{\Ec\to \Ec} \leq e^{2\beta B+1}\cbeta$. Lastly, for some bounded set $\Lambda\subset\R^d$ let $\bx\colon\Ec \to \Ec$ be the projection operator
\begin{align*}
    \bx\colon  &\Ec \to \Ec\\
    &\bo \mapsto \bx\bo = (\mathds{1}_{\Lambda^n}\omega^{(n)})_{n\geq 1},
\end{align*}
$\II\colon \Ec \to \Ec$ be the identity and $\ee=(e_1^{(n)})_{n\geq 1}$ be the vector in $\Ec$ with $e^{(1)}_1\equiv 1$ and $e^{(n)}_1\equiv 0$ for $n\geq 2$. \\
For a given $z\in\C$ and bounded $\Lambda\subset \R^d$ consider the \emph{finite volume Kirkwood-Salsburg equations} defined by
\begin{align}\label{eq:fvKS}
    (\II-z\bx\BPi\K)\bo = z\bx\ee.
\end{align}
In the context of statistical mechanics $z$ (usually $z>0$) is called the \emph{activity} of the \emph{grand-canonical ensemble} associated to $(\beta,z,u)$. 
It is well-known that for $z\in\Bz:= \{z\in\C \mid |z|<z_0\}$ where 
\begin{align}\label{eq:gasphase}
    z_0:= \left( e^{2\beta B+1}\cbeta \right)^{-1}
\end{align}
there is a unique solution to \req{fvKS} which can be developed into a Neumann-series, i.e. the solution is given by
\begin{align}\label{eq:fvKSsol}
    \bt_\Lambda(z) = (\II-z\bx\BPi\K)^{-1} z\bx\ee = \sum_{k=0}^\infty (z\bx\BPi\K)^kz\bx\ee.
\end{align}
In particular this means that for each $n$ and $x_1,\dots,x_n\in\Lambda$ the function $\theta_\Lambda^{(n)}(z;\xx_n) $ is an analytic function on $B_{z_0}$.
Furthermore, the solution of \req{fvKS} can be written down explicitly using the \emph{grand canonical partition function}
\begin{align}\label{eq:partfct}
    \Xi_\Lambda(z)= 1 +\sum_{k=1}^\infty\frac{z^k}{k!}
    \int_{\Lambda^k} e^{-\beta H(\{\yy_k\})}\dyy_k.
\end{align}
As shown by Ruelle, see \cite{Ruelle69}, $\Xi_\Lambda(z)\neq 0$ for $z\in\Bz$, which implies that 
\begin{align}\label{eq:solxpl}
    \theta_\Lambda^{(n)}(z;\xx_n)= \frac{1}{\Xi_\Lambda(z)}
    \sum_{k=0}^\infty\frac{z^{n+k}}{k!}
    \int_{\Lambda^k} e^{-\beta H(\{\xx_n,\yy_k\})}\dyy_k.
\end{align}
\begin{remark}
From the proof of Theorem \ref{thm:mainthm} one will see that the Janossy densities of the Kirkwood closure process $\KK_{\varsigma,\phi}$ for $\varsigma=z$ and $\phi=e^{-\beta u}$ are given by
\begin{align}\label{eq:jask}
    j_\Lambda^{(n)}(z;\xx_n)=(-1)^n\Xi_\Lambda(-z)\theta^{(n)}_\Lambda(-z;\xx_n).
\end{align}
In particular, the probability of finding no points in a given bounded set $\Lambda\subset\R^d$ is given by 
\begin{align*}
    \KK_{z,e^{-\beta u}}(N_\Lambda = 0 ) = \Xi_\Lambda(-z) \quad 
    \text{ for all } z \in (0,z_0).
\end{align*}
This resembles results about non-vanishing probabilities in statistical mechanics, see e.g.~the fundamental theorem in \cite{Scott06} for the case of lattice gases.    
\end{remark}\\
\begin{remark}\label{rem:KSnoswitch}
Note that the solution $\bt_\Lambda$ of \req{fvKS} also satisfies the Kirkwood-Salsburg equation without the permutation operator $\BPi$, namely,
\begin{align*}
    (\II-z\bx\K)\bt_\Lambda = z\bx\ee
\end{align*}
by construction. 
\end{remark}\\
The argument that $\Xi_\Lambda(z)\neq 0$ by Ruelle is as follows: For $z>0$ and $n=1$ integration of \req{solxpl} with respect to $x$ and differentiation of \req{partfct} with respect to $z$ shows that 
\begin{align}\label{eq:anal}
	\int_{\Lambda}\theta^{(1)}_\Lambda(z;x)\dx =z \frac{\diff}{\diff z}\log \Xi_\Lambda(z).
\end{align}
Since by \req{fvKSsol} the left-hand side is analytic in $\Bz$  this implies that the right-hand can also be continued as an analytic function, meaning $\Xi_\Lambda(z)$ does not have any zeros in $\Bz$. Using a similar argument Kuna, Lebowitz and Speer, see \cite{Kuna07}, to prove the existence of the Kirkwood closure process for \emph{locally stable} interactions. This will be elaborated on in Subsection \ref{ss:locstab}.\\
To conclude this section some more properties of the solutions of \req{fvKS} will be stated. It follows from \req{fvKSsol} that the solutions $(\theta_\Lambda^{(n)}(z;\cdot))_{n\geq 1}$ satisfy
\begin{align}\label{eq:fvStable}
    \left| \theta_\Lambda^{(n)}(z;\xx_n)\right| 
    \leq 
    \left(\frac{1}{\cbeta}\max\left\lbrace 
    \frac{\cbeta |z|}{1-|z|/z_0},1 \right\rbrace\right)^n.
\end{align}
This bound is independent of $\Lambda$ and it can be shown that when choosing a sequence of increasing sets $\Lambda_l\subset \Lambda_{l+1}$ such that for any bounded set $\Delta\subset \R^d$ there is an $l_0$ such that $\Delta \subset \Lambda_{l_0}$ (this limit is denoted by $\Lambda\nearrow\R^d$) the solutions of \req{fvKS} converge in the weak$*$ topology to some $\bt = (\theta^{(n)})_{n\geq 1}$, i.e.
\begin{align}\label{eq:weakstar}
    \lim_{\Lambda\nearrow\R^d} \left|\int_{(\R^d)^n }F(\xx_n)\,\theta^{(n)}(z;\xx_n)\dxx_n -\int_{\Lambda^n }F(\xx_n)\,\theta^{(n)}_\Lambda(z;\xx_n)\dxx_n\right|=0
\end{align}
for any $n\geq 1$ and $F\in L^1((\R^d)^n)$ which is the unique solution of the \emph{infinite volume Kirkwood-Salsburg equations}
\begin{align}\label{eq:KS}
    (\II-z\BPi\K)\bt = z\ee.
\end{align}
For $z>0 $ the solutions $(\theta^{(n)}_\Lambda)_{n\geq 1}$ of the finite volume Kirkwood-Salsburg equations \req{fvKS} are the correlation functions of the so-called \emph{grand canonical Gibbs measure} $\GG_{\Lambda,\beta,z,u}$ on $\Lambda$. It can be shown  finite volume Gibbs measures converge to a limit $\PP_{\beta,z,u}$, cf.~\cite{Ruelle70}. This limit is tempered and satisfies the \emph{(multivariate) GNZ-equation} (named for Georgii, Nguyen and Zessin), i.e.~for every $F\colon (\R^d)^n\times \Gamma\to [0,+\infty]$ there holds
\begin{align}\label{eq:GNZeq}
    \int_{\Gamma}\sum_{x_1,\dots,x_n\in\eta} F(\xx_n;\eta) \dP_{\beta,z,u}
    =\int_{(\R^d)^n}\int_\Gamma F(\xx_n;\eta\cup \{\xx_n\}) z^ne^{-\beta H(\xx_n)-\beta W(\{\xx_n\}\mid \eta)} 
    \dP_{\beta,z,u} \dxx_n.
\end{align}
Thus, $\PP_{\beta,z,u}$ is a so-called \emph{$(\beta,z,u)$-Gibbs measure} and the correlation functions of $\PP_{\beta,z,u}$ solve \req{KS}. The function
\begin{align}\label{eq:papa}
    \kappa_{\beta,z,u}(\xx_n;\eta) := z^ne^{-\beta H(\xx_n)-\beta W(\{\xx_n\}\mid \eta)} 
\end{align}
is also called a \emph{Papangelou kernel.} 
\\
\begin{remark}\label{rem:conv}
It is shown in \cite{Ruelle70} that when taking the limit $\Lambda\nearrow\R^d$ the associated solutions of \req{fvKS} also converge uniformly on compacts to the solution of \req{KS}, i.e.~for any $n\geq 1$, $\Delta\subset \R^d$ compact there holds
\begin{align}\label{eq:unifcpts}
    \lim_{\Lambda \nearrow\R^d} \sup_{\xx_n\in\Delta^n}
    \left|\theta^{(n)}(z;\xx_n)-\theta^{(n)}_\Lambda(z;\xx_n)\right| = 0.
\end{align}
Thus, \req{unifcpts} and Remark \ref{rem:KSnoswitch} imply that the solution $\bt$ of \req{KS} also satisfies
\begin{align}\label{eq:KSnoswitch}
    (\II-z\K)\bt = z\ee.
\end{align}
In other words, there holds 
\begin{align}\label{eq:ksns}
    \theta^{(n+1)}(z;x,\xx_{n})
    =z
    e^{-\beta W(\{x \}\mid \{\xx_{n}\})}
    \sum_{k=0}^\infty\frac{1}{k!}\int_{(\R^d)^k}
    \prod_{j=1}^k \mayer(x-y_j)
    \theta^{(n+k)}(z;\xx_{n},\yy_k)\dyy_k
    .
\end{align}
\end{remark}\\\noindent
Lastly, some dualities between the solutions of \req{fvKS} and \req{KS} for $z\in \R$ are noted. 
\begin{itemize}
    \item $z>0$: ${\theta}_\Lambda^{(n)}(z,\cdot)= \rho^{(n)}_{\Lambda}$ are the correlation functions of the {grand canonical Gibbs measure} $\GG_{\Lambda,\beta,z,u}$ on $\Lambda$ and thus the underlying measure is a \emph{different} measure for different sets $\Lambda$ and $\Lambda'$. In the limit $\Lambda \nearrow\R^d$ these correlation functions converge to the solution of \req{KS}, i.e.~the correlation functions of the infinite volume measure $\PP_{\beta, z,u}$. Since the Hamiltonian associated to $u$ is stable these correlation functions satisfy Ruelle's bound by virtue of \req{fvStable}.
    \item $z<0$: ${\theta}_\Lambda^{(n)}(z,\cdot)= (-1)^n j^{(n)}_\Lambda/j_\Lambda^{(0)}$ is a quotient of Janossy densities of the \emph{same} underlying point process (which is the Kirkwood closure process). Heuristically, one can interpret this quotient as a so-called \emph{Boltzmann factor}, i.e. there is some Hamiltonian $H_\Lambda$ such that 
    \begin{align*}
        \frac{j^{(n)}_\Lambda}{j_\Lambda^{(0)}} = e^{-H_\Lambda}.
    \end{align*}
    This Hamiltonian is stable by virtue of \req{fvStable} and depends on the set $\Lambda$ since the Janossy densities contain averaged information of the outside of $\Lambda$. In the same way as for $z>0$ one can expect that $e^{-H_\Lambda}$ converges to some Hamiltonian $H_\KK$ for which the Kirkwood closure process is Gibbs, as previously mentioned, this will be discussed in Section \ref{sec:GNZresult}.
\end{itemize}
The above discussion motivates the definition of  
the Hamiltonian $H_\KK$ by
\begin{align}\label{eq:kirkwoodH}
    H_\KK(z;\{\xx_n\}) := -\log\iota^{(n)}(z;\xx_n)
\end{align}
with
\begin{align}\label{eq:iota}
    \iota^{(n)}(z;\xx_n):= (-1)^n \theta^{(n)}(-z;\xx_n) , \qquad \xx_n \in (\R^d)^n.
\end{align}
This Hamiltonian is stable because it follows from \req{fvStable} and \req{unifcpts} that 
\begin{align}
    0\leq \iota^{(n)}(z;\xx_n) \leq \left(\frac{1}{\cbeta}\max\left\lbrace \frac{\cbeta |z|}{1-|z|/z_0},1 \right\rbrace\right)^n.
\end{align}
Furthermore, the Hamiltonian includes a non-trivial one-body term, i.e. the activity, given by $\iota^{(1)}$ as each entry of the unique solution to $\bt$ of \req{KS} is invariant under translations of its arguments.

\subsection{Locally stable interactions}\label{ss:locstab}
The local stability condition gives a lot more control over the interaction. In particular, the permutation operator $\BPi$ is not needed to ensure the Kirkwood-Salsburg operator is an endomorphism and boundary conditions for the Kirkwood-Salsburg equations can be introduced. Let $\nu$ be a measure on $(\Gamma,\mathscr{F})$ with $\nu( \Gamma \backslash \Gamma_*) = 0$ and define the spaces
\begin{align*}
    E_{\zeta,\nu}^1:= L^1(\Gamma_0\times\Gamma)=
    \left\lbrace
    \boldsymbol{F}= (F^{(n)})_{n\geq 1}
    \;\middle|\;
    F^{(n)}\colon (\R^d)^{n}\times \Gamma\to \C, \,\,\,
    \|\boldsymbol{F}\|_{1,\nu}<+\infty
    \right\rbrace
\end{align*}
where 
\begin{align*}
    \|\boldsymbol{F}\|_{1,\nu} := 
    \sum_{n=1}^\infty
    \frac{\zeta^{-n}}{n!}
    \int_{(\R^d)^n}\int_\Gamma 
    \left|F^{(n)}(\xx_n;\eta) \right| \diff\nu \dxx_n
\end{align*}
and
\begin{align*}
    E_{\zeta,\nu}^\infty:= L^\infty(\Gamma_0\times\Gamma)=
    \left\lbrace
    \bo= (\omega^{(n)})_{n\geq 1}
    \;\middle|\;
    \omega^{(n)}\colon (\R^d)^{n}\times \Gamma\to \C, \,\,\,
    \|\bo\|_{\infty,\nu}<+\infty
    \right\rbrace
\end{align*}
where 
\begin{align*}
    \|\bo\|_{\infty,\nu}:=
    \sup_{n\geq 1}\left(\zeta^n\!\!\!
    \esssup_{(\xx_n;\eta)\in (\R^d)^n\times\Gamma}\!|\omega^{(n)}(\xx_n;\eta)|\right)
\end{align*}
and the essential supremum is taken with respect to $ \lebesgue^n\times\nu$. Define the operator $\K_{\Gamma} \colon E_{\cbeta,\nu}^\infty \to E_{\cbeta,\nu}^\infty $ by 
\begin{align}\label{eq:KGammadef1}
    (\K_{\Gamma}\bo)^{(1)}(x;\eta)=
    e^{-\beta W(\{x \}\mid \eta)}\sum_{k=1}^\infty\frac{1}{k!}\int_{(\R^d)^k}
    \prod_{i=1}^k \mayer(x-y_i)
    \omega^{(k)}(\yy_k;\eta)\dyy_k
\end{align}
and for $n\geq 1$ by
\begin{equation}\label{eq:KGammadef2}
\begin{split}
    (\K_{\Gamma}\bo)^{(n+1)}(x,\xx_n;\eta)=  
    \,e^{-\beta W(\{x\} \mid \eta\cup \{\xx_n\})}
    \sum_{k=0}^\infty\frac{1}{k!}\int_{(\R^d)^k}
    \prod_{j=1}^k \mayer(x-y_j)
    \omega^{(n+k)}(\xx_n,\yy_k;\eta)\dyy_k.
\end{split}
\end{equation}
By  \req{winteraction} the terms $e^{-\beta W(\{x \}\mid \eta)}$ and $e^{-\beta W(\{x\} \mid \eta\cup \{\xx_n\})}$ are well-defined and as in Subsection \ref{ss:KSOP}, $\|\K_{\Gamma}\|_{E_{\nu}^\infty \to E_{\nu}^\infty} \leq  e^{2\beta B+1}\cbeta $ and thus $\K_{\Gamma}$ is well-defined. 
Further, denote by $\II\colon E_{\cbeta,\nu}^\infty\to E_{\cbeta,\nu}^\infty$ the identity operator, and for $\Lambda\subset \R^d$ by $\bx\colon  E_{\cbeta,\nu}^\infty \to E_{\cbeta,\nu}^\infty$ the projection operator 
\begin{align*}
    \bx\colon  &E_{\cbeta,\nu}^\infty \to E_{\cbeta,\nu}^\infty \\
    &\bo \mapsto \bx\bo = (\mathds{1}_{ \Lambda^n}\omega^{(n)}(\,\cdot\,; \,\cdot\cap\Lambda\,)_{n\geq 1}.
\end{align*}
As $\|\K_{\Gamma}\|_{E_{\cbeta,\nu}^\infty \to E_{\cbeta,\nu}^\infty} \leq  e^{2\beta B+1}\cbeta $ the operators $\II-z\bx\K_{\Gamma}$ and $\II-z\K_{\Gamma}$ are invertible for every $z \in \Bz$.
In particular, the equations 
\begin{align}\label{eq:fvKNS}
    (\II-z\bx\K_{\Gamma})\bo = z\bx \al
\end{align}
and 
\begin{align}\label{eq:KNS}
    (\II-z\K_{\Gamma})\bo = z \al
\end{align}
have unique solutions $\btt_{\Lambda}(z)= (\vartheta_{\Lambda}^{(n)}(z;\,\cdot\,;\,\cdot\,))_{n\geq 1}$ and $\btt(z)=(\vartheta^{(n)}(z;\,\cdot\,;\,\cdot\,))_{n\geq 1}$ in $E^\infty_\Gamma$ for any right-hand side $\al\in E_{\nu}^\infty$.
\begin{remark}\label{rem:kuna}
To recover the results of Kuna et al. from \cite{Kuna07} let $\widetilde{\xx}_l =\{\widetilde{x}_1,\dots,\widetilde{x}_l\} $ for some $l\in \N$ and $\widetilde{x}_1,\dots,\widetilde{x}_l\in \Lambda$ and take $\nu = \delta_{\widetilde{\xx}_l}$ and $\al =\ee$. It is easy to see that the solution $\btt_{\Lambda}(z)$ of \req{fvKNS} is given by
\begin{align*}
    \vartheta_\Lambda^{(n)}(z;\xx_n ;\widetilde{\xx}_l) = 
    \frac{1}{\Xi_\Lambda (z;\widetilde{\xx}_l)}
    \sum_{k=0}^\infty 
    \frac{z^{n+k}}{k!}
    e^{-\beta H(\{\xx_n,\yy_k\})-\beta W(\{\xx_n,\yy_k\}\mid \widetilde{\xx}_l)} \dyy_k
\end{align*}
where 
\begin{align*}
    \Xi_\Lambda (z;\widetilde{\xx}_l) = 1 + \sum_{k=1}^\infty
    \frac{z^k}{k!} \int_{\Lambda^k}
    e^{-\beta H(\{\yy_k\})-\beta W(\{\yy_k\}\mid \widetilde{\xx}_l)} \dyy_k.
\end{align*}
Using the arguments of Ruelle they conclude $\Xi_\Lambda (z;\widetilde{\xx}_l)\neq 0$ in $\Bz$. Lastly, one can observe that the Janossy densities of the Kirkwood closure with $\varsigma=z$ and $\phi=e^{-\beta  u}$ are given by
\begin{align}\label{eq:mayereqs}
    j_\Lambda^{(n)} (\xx_n) = z^n e^{-\beta H(\{{\xx}_n\})}\Xi_\Lambda (-z; \{\xx_n\}).
\end{align}
It is easy to prove that  
\begin{align*}
    \Xi_\Lambda (-z; \{\xx_n\})
    =1+ \sum_{k=1}^\infty
    \frac{(-1)^k}{k!}
    \int_{\Lambda^k}
    \prod_{j=1}^k\left(e^{-\beta W(\{y_j\}\mid \{\xx_n\})}-1\right)j^{(k)}_\Lambda(\yy_k) \dyy_k
\end{align*}
and thus \req{mayereqs} can be seen a version of the \emph{Mayer-Montroll equation}.
\end{remark}\\
Comparison of \req{mayereqs} and \req{jask} reveals that 
\begin{align*}
    \theta^{(n)}_\Lambda(-z;\xx_n)
    = \frac{(-z)^n e^{-\beta H(\{{\xx}_n\})}\Xi_\Lambda (-z; \{\xx_n\})}{\Xi_\Lambda (-z)} .
\end{align*}
Since $ \Xi_\Lambda (-z; \{\xx_n\})/\Xi_\Lambda (-z)\neq 0$ in $\Bz$ one can conclude that:\noindent
\begin{corollary}\label{corr:pos}
Let $u$ be locally stable and regular, then for any $n\in\N$ and $x_1,\dots,x_n\in \Lambda$ the fraction $\theta^{(n)}(z;\xx_n)/z^n$ is either positive or equal to zero in $\Bz.$ 
\end{corollary}
\begin{remark}\label{rem:hereditary}
Note that $\theta^{(n)}_\Lambda(z;\xx_n)=0$ for some $\xx_n \in \Lambda^n$ also implies $\theta^{(n+k)}_\Lambda(z;\xx_n,\yy_k)=0$ for all $k\geq 1$ and any $\yy_k\in\Lambda^k$ by Corollary \ref{corr:pos}, meaning $\bt$ inherits the hereditarity of the Hamiltonian $H$. 
\end{remark}

\section{Existence of the Kirkwood closure process}\label{sec:results}
The main result can now be stated. 
\begin{theorem}\label{thm:mainthm}
Let $\beta>0$, $z\in(0,z_0)$ (with $z_0$ as in \req{gasphase}) and $u\colon \R^d \to \R\cup\{+\infty\}$ be a stable and regular pair interaction. For $\varsigma=z$ and $\phi= e^{-\beta u}$ the Kirkwood closure process $\KK_{\varsigma,\phi}$ exists and is tempered.
\end{theorem}
\begin{remark}
Since the correlation functions of $\KK_{\varsigma,\phi}$ satisfy Ruelle's bound for $\xi = z e^{\beta B}$ by construction, it follows that $\KK_{\varsigma,\phi}$ is tempered. 
\end{remark}\noindent\\
As previously mentioned, in computational physics the Kirkwood superposition approximation is used to approximate the correlation functions of Gibbs measures. Theorem \ref{thm:mainthm} can be used to establish an existence result of the corresponding Kirkwood closure process under some additional decay assumptions on the pair potential.
\begin{corollary}
Let $\beta,z>0$, $u\colon\R^d\to \R\cup\{+\infty\}$ be of Lennard-Jones type, and $\PP_{\beta,z,u}$ be a corresponding $(\beta,z,u)$-Gibbs measure with density $\rho$ and radial distribution function $g$. If $z$ is sufficiently small, the Kirkwood closure process $\KK_{\rho,g}$ for the pair $(\rho,g)$ exists.
\end{corollary}
\begin{proof}
It is well-known, cf.~\cite{Ruelle69}, that $\rho=\rho(z)$ is a decreasing function of $z$. Furthermore, for $z$ sufficiently small it was shown in \cite{Hanke18c} that there exists a Lennard-Jones type potential $v$ such that $g=e^{- v}$. Therefore, if $z$ is small enough, so is $\rho$ and the Kirkwood closure for $(\rho,g)$ exists by Theorem \ref{thm:mainthm}.
\end{proof}\\ \noindent

\begin{proposition}\label{prop:pos}
For any $x_1,\dots ,x_n\in\Lambda$ and any $z\in (0,z_0)$ there holds
\begin{align}\label{eq:altsign}
    (-1)^n \theta_\Lambda^{(n)}(-z;\xx_n) \geq 0.
\end{align}
\end{proposition}
The idea of the proof is to approximate the potential $u$ by an appropriate potential $u_\delta$ that is locally stable and show that the corresponding solutions of \req{fvKS} converge in the weak$*$ topology for $\delta\to 0$. 
\begin{proof}
For a given $z\in (0,z_0)$ choose $\delta>0$ such that $z\in (0,z_\delta)$ where
\begin{align*}
	z_\delta := \left( e^{2\beta B+1}\exp\left(\delta/\cbeta\right)\cbeta) \right)^{-1} 
\end{align*}
and define 
\begin{align}\label{eq:udelta}
		u_{\delta}: = u + \infty\cdot \mathds{1}_{|x|< r_0}.
\end{align}
Here $r_0=r_0(\delta)>0$ is chosen such that
\begin{align*}
	\int_{\R^d} \left|\mayer^{\delta}(x)\right|\dx \leq \cbeta+\delta
\end{align*}
where $\mayer^{\delta}: = e^{-\beta u_\delta(\cdot )}-1$.
In particular, since $u_\delta \geq u$ one can use the same stability constant for $u_\delta$ as for $u$. Furthermore, if \req{locstab} holds for $u$ it also holds for $u_\delta$ and thus the definition of $\PI$ does not need to be changed as one can use the index $i_*$ defined by $u$ for every $u_\delta$. One can now define $\K_\delta$ as in \req{Kdef1} and \req{Kdef2} with $u_\delta$ in place of $u$ and gets that the corresponding version of \req{fvKS} has a unique solution $\bt_{\Lambda,\delta}$ for $|z|<z_\delta$ since 
$\|\BPi\K_\delta\|_{\Ec\to \Ec} \leq e^{2\beta B+1}\exp\left(\delta/\cbeta\right)\cbeta$. Since every $u_\delta$ is locally stable it follows from Corollary \ref{corr:pos} that for every $n$ and $x_1,\dots ,x_n\in \Lambda$ there holds
\begin{align}\label{eq:altsigndelta}
	 \textnormal{sgn}(z)^n \theta_{\Lambda,\delta}^{(n)}(z;\xx_n) >0 
    \quad\text{ for all } z \in \Bz\cap\R\backslash\{0\}.
\end{align}
%
From \req{fvKSsol} it follows that
\begin{align}\label{eq:KSdelta}
	\|\bt_{\Lambda,\delta}\|_{\cbeta} \leq 
	|z|\cbeta 
	\sum_{k=0}^\infty |z|^k\|\bx\BPi\K_\delta\|_{\Ec\to \Ec}^k
	=\frac{z\cbeta }{1-|z| e^{2\beta B+1}\exp\left(\delta/\cbeta\right)\cbeta }
\end{align}
and it follows that the sequence $(\bt_{\Lambda,\delta})_{\delta>0}$ has a subsequence for $\delta\to 0$ such that for every $n\geq 0$ and $F_{n+1}\in L^1((\R^d)^{n+1})$ there holds
\begin{align}\label{eq:deltaconv}
	\lim_{\delta\to 0}
	\int_{(\R^d)^{n+1}}
	F_{n+1}(x,\xx_n)\theta_{\Lambda,\delta}^{(n+1)}(x,\xx_n)\diff(x,\xx_n)
	=\int_{(\R^d)^n}
	F_{n+1}(x,\xx_n)\theta_{\Lambda,*}^{(n+1)}(x,\xx_n)\diff(x,\xx_n)
\end{align} 
for some $\bt_{\Lambda,*}=(\theta^{(n)}_{\Lambda,*})_{n\geq 1}$. Since $\bt_{\Lambda,\delta}$ satisfies the Kirkwood-Salsburg equations for the modified potential $u_\delta$ one can conclude that
\begin{align*}
&\int_{(\R^d)^{n+1}}
	F_{n+1}(x,\xx_n)\theta_{\Lambda,\delta}^{(n+1)}(x,\xx_n)\diff(x,\xx_n) \\
	&=\int_{(\R^d)^n}
	F_{n+1}(x,\xx_n)z
    e^{-\beta W_\delta(\{x \}\mid \{\xx_{n}\})}
    \left(
    \theta_{\Lambda,\delta}^{(n)}(\xx_{n})+
    \sum_{k=1}^\infty\frac{(-1)^k}{k!}\int_{(\R^d)^k}
    \prod_{j=1}^k \mayer^{\delta}(x-y_j)
    \theta_{\Lambda,\delta}^{(n+k)}(z;\xx_{n},\yy_k)\dyy_k
    \right).
\end{align*}
again with the convention $\theta_{\Lambda,\delta}^{(0)}:=1$.
Define 
\begin{align}\label{eq:ftilde}
	\widetilde{F}^{\delta}_{n+1+k}(x,\xx_n,\yy_k):=
	F_{n+1}(x,\xx_n)z
    e^{-\beta W_\delta(\{x \}\mid \{\xx_{n}\})}
    \prod_{j=1}^k \mayer^{\delta}(x-y_j)
\end{align}
then there holds 
\begin{align*}
	\lim_{\delta \to 0}
	\widetilde{F}^{\delta}_{n+1+k}(x,\xx_n,\yy_k)
	=F_{n+1}(x,\xx_n)z
    e^{-\beta W(\{x \}\mid \{\xx_{n}\})}
    \prod_{j=1}^k \mayer(x-y_j)
\end{align*}
almost everywhere and furthermore by \req{locstab} for $\delta$ small enough there holds
\begin{align}
	\left|
	F_{n+1}(x,\xx_n)z
    e^{-\beta W_\delta(\{x \}\mid \{\xx_{n}\})}
    \prod_{j=1}^k \mayer^{\delta}(x-y_j)
	\right|\leq
	\left|
	F_{n+1}(x,\xx_n)z
    e^{2\beta B}
    \prod_{j=1}^k \left(\mathds{1}_{|x-y_j|< 1}+\mayer(x-y_j)\right)
	\right|.
\end{align}
Using dominated convergence it can thus be concluded that 
\begin{align}
	\lim_{\delta \to 0}
	\int_{(\R^d)^{n+1+k}}\left|
	\widetilde{F}^{\delta}_{n+1+k}(x,\xx_n,\yy_k)
	-F_{n+1}(x,\xx_n)z
    e^{-\beta W(\{x \}\mid \{\xx_{n}\})}
    \prod_{j=1}^k \mayer(x-y_j)
	\right|\diff(x,\xx_n,\yy_k) =0.
\end{align}
The $L^1$ convergence of the $\widetilde{F}^{\delta}_{n+1+k}$ together with \req{deltaconv} then shows that the limit $\bt_{\Lambda,*}$ satisfies the Kirkwood-Salsburg equations for the original potential $u$ and satisfies \req{altsign} since every $\bt_{\Lambda,\delta} $ satisfies \req{altsigndelta}. Since this solution is unique it follows $\bt_{\Lambda,*}=\bt_{\Lambda}$ and the proposition is proved.
\end{proof}

\begin{proofthm}[\ref{thm:mainthm}]
Let $u\colon \R^d\to \R\cup \{+\infty\}$ be a regular and stable pair potential and $z \in (0,z_0)$. Let $(\rho^{(n)})_{n\geq 1}$ be defined by \req{easycorrelations} with $\varsigma = z$ and $\phi = e^{-\beta u}$. Then the functions $(\rho^{(n)})_{n\geq 1}$ satisfy Ruelle's bound \req{Rbound} with $\xi= z e^{\beta B}$. It remains to show that the inequalities \req{lenpos1} and \req{lenpos2} are satisfied for every bounded $\Lambda\subset \R^d$. For \req{lenpos2} this follow immediately as
\begin{align*}
    1+\sum_{k=1}^\infty\frac{(-1)^k}{k!} \int_{\Lambda^k}\rho^{(k)}(\yy_k)\dyy_k
    =\Xi_\Lambda (-z)
\end{align*}
and $\Xz$ has no zeros in $\Bz$. Finally, let $\sigma_\Lambda^{(0)}(z):= \Xi_\Lambda(-z)$ and for $n\geq 1$
\begin{align*}
    \sigma_\Lambda^{(n)}(z;\xx_n):= (-1)^n \Xi_\Lambda(-z) \theta_\Lambda^{(n)}(-z;\xx_n)
\end{align*}
where $(\theta_\Lambda^{(n)}(-z;\cdot))_{n\geq 1}$ is the solution of \req{fvKS} for $-z$. By Proposition \ref{prop:pos} $\sigma_\Lambda^{(n)}(z;\xx_n) \geq 0$ for all $\xx_n\in\Lambda^n$ and since
\begin{align*}
    \sigma_\Lambda^{(n)}(z;\xx_n) = 
    \sum_{k=0}^\infty\frac{(-1)^k}{k!} \int_{\Lambda^k}\rho^{(n+k)}(\xx_n,\yy_k)\dyy_k
\end{align*}
by virtue of \req{solxpl}, the theorem is proved. 
\end{proofthm}

\begin{remark}
Theorem \ref{thm:mainthm} can also be extended to the case that the pair potential is not translationally invariant, i.e.~$u\colon (\R^d)^2\to \R\cup \{+\infty\}$. The proof works the same way, however, some additional technical assumptions on $u$ need to be made, cf.~\cite{KunaPhD}
\end{remark}

\section{The Kirkwood closure process is a  Gibbs point process}\label{sec:GNZresult}
In this section it will be shown that for locally stable $u$ the Papangelou kernel of the Kirkwood closure process $\KK_{\varsigma,\phi}$ for $\varsigma=z$ and $\phi=e^{-\beta u}$ solves a modified Kirkwood-Salsburg equation.
In particular, it is shown that $\KK_{\varsigma,\phi}$ is a Gibbs point process for the Hamiltonian defined in \req{kirkwoodH}. 
For finite configurations the interaction $W_\KK$ associated to $H_\KK$ is characterized by \req{intprop}, i.e.
\begin{align}\label{eq:WKprop}
    H_\KK(z;\eta\cup \gamma) = 
    H_\KK(z;\gamma)+
    W_\KK(z;\gamma \mid \eta)+
    H_\KK(z;\eta)
\end{align}
for $\eta,\gamma \in \Gamma_0$. Using \req{WKprop} and \req{kirkwoodH} it can be concluded that for $\xx_n\in (\R^d)^n$ and $\eta \in \Gamma_0$ there holds
\begin{align}\label{eq:wasiota}
    \frac{\iota^{(N(\eta)+n)}(z;\xx_n,\eta)}{\iota^{(N(\eta))}(z;\eta)}
    = e^{- H_\KK(z;\{\xx_n\})- 
    W_\KK(z;\{\xx_n\} \mid \eta)}
\end{align}
where by abuse of notation $\eta $ in the argument of $\iota^{(n+N(\eta))}$ (respectively $\iota^{(N(\eta))}$) denotes the vector containing the points of $\eta.$ This fraction is well-defined by Corollary \ref{corr:pos}.
Thus \req{wasiota} can be used to define a Papangelou kernel $\kappa$ (analogous to \req{papa}) of the Kirkwood closure process for finite configurations as
\begin{align*}
    \kappa^{(n)}(z;\xx_n;\eta):=\frac{\iota^{(n+N(\eta))}(z;\xx_n,\eta)}{\iota^{(N(\eta))}(z;\eta)}.
\end{align*}
However, since the Kirkwood closure process is translationally invariant there holds $\KK_{\varsigma,\phi}(N_{\R^d}(\eta)<+\infty)=0$ and thus the \glqq typical\grqq{} $\eta$ will have infinitely many points and a way to define the interaction $W_\KK$ (and thus the kernel $\kappa$) for infinite $\eta$ is needed. Defining $\vartheta^{(n)}(z;\xx_n;\eta):=(-1)^n\kappa^{(n)}(-z;\xx_n;\eta)$ and plugging \req{wasiota} into \req{ksns} one finds that
\begin{align}\label{eq:gnzkirkwoodinf}
    \vartheta^{(n)} (z; \xx_{n};\eta)
    = ze^{-\beta W(x_1\mid \{\xx_{2,n}\}\cup \eta)}  
    \sum_{k=0}^\infty 
    \frac{1}{k!}
    \int_{(\R^d)^k} \prod_{i=1}^k f_\beta(x_1-y_i)
    \vartheta^{(n+k-1)} (z; \xx_{2,n},\yy_k;\eta)\dyy_k
\end{align}
with the convention $\vartheta_\Lambda^{(0)}(\eta)\equiv 1$. To gain control of the limit $\Lambda\nearrow\R^d$ of interactions restricted to a bounded set $\Lambda$ an additional assumption is made.\\
A pair potential $u$ is called \emph{lower regular}, i.e.~there exists a decreasing function $\psi\colon [0,+\infty)\to [0,+\infty)$ with 
    \begin{align*}
        \int_0^\infty \psi(r)r^{d-1}\dr < +\infty
    \end{align*}
    and for all $x\in\R^d$
    \begin{align*}
        u(x) \geq - \psi(|x|) .
    \end{align*}
For stable and lower regular pair potentials $u$ and $W$ defined by \req{winteraction} it is known that for any tempered point process $\PP$ and $\eta\in\Gamma_0$ there holds 
\begin{align}\label{eq:wcont}
    W(\eta\mid \gamma) = \lim_{\Lambda\nearrow\R^d}W(\eta\mid \gamma_\Lambda) \in \R\cup\{+\infty\}, 
\end{align}
for $\PP$-almost all $\gamma\in\Gamma,$ see \cite{KunaPhD}.
Since the Kirkwood closure process is tempered by Theorem \ref{thm:mainthm} equation \req{gnzkirkwoodinf} is well-defined for $ \lebesgue^n\times\KK_{\rho,g}$-almost all $(\xx_n,\eta)\in (\R^d)^n\times \Gamma$ and every $n\geq 1$ and can be used to define $\kappa^{(n)}$, it remains to show that $(\kappa^{(n)})_{n\geq 1}$ is indeed the Papangelou kernel of the Kirkwood-closure process.
\begin{theorem}\label{thm:thm3}
Let $\beta>0,$ $z\in(0,z_0)$ and $u\colon\R^d \to \R\cup\{+\infty\}$ be locally stable, regular and lower regular, and let $\KK_{\varsigma,\phi}$ be the Kirkwood closure process for $\varsigma=z$ and $\phi=e^{-\beta u}$, then for any $n\geq 1$ and any nonnegative function $F\colon (\R^d)^n\times \Gamma\to [0,+\infty)$ there holds
\begin{align}\label{eq:KirkwoodGNZ}
    \int_\Gamma 
    \sum_{x_1,\dots,x_n\in\eta \atop x_i \neq x_j }
    F(\xx_n;\eta) 
    \dK_{\varsigma,\phi}(\eta)
    =\int_{(\R^d)^n} \int_\Gamma 
    F(\xx_n;\eta\cup\{\xx_n\})
    \kappa^{(n)}(z;\xx_n;\eta) \dK_{\varsigma,\phi}(\eta) \dxx_n
\end{align}
where $(-1)^n\kappa^{(n)}(-z;\,\cdot\,;\,\cdot\,)\colon  (\R^d)^n\times \Gamma\to [0,+\infty)$ solves \req{gnzkirkwoodinf}.
In particular, $\KK_{\varsigma,\phi}$ satisfies the multivariate GNZ-equation and is thus a $H_\KK$-Gibbs measure for $H_\KK$ given by \req{kirkwoodH}.
\end{theorem}
\begin{remark}
Since the Janossy densities of the Kirkwood closure process are given by \req{jask} it follows from \req{unifcpts} that for all compact $\Delta\subset \R^d$ there holds
\begin{align*}
    \lim_{\Lambda \nearrow\R^d} \sup_{\xx_n\in\Delta^n}
    \left|
    \frac{j_\Lambda^{(n)}(z;\xx_n)}{j_\Lambda^{(0)}(z)}
    -e^{-H_\KK(\{\xx_n\})}\right| = 0.
\end{align*}
\end{remark}
In light of this, one can first look at the restriction of the Kirkwood closure process $\KK_{\varsigma,\phi}$ to a finite volume. However, since this convergence only holds for finite configurations one needs to be careful when taking the limit.

\begin{lemma}
Let $\beta>0$, $z\in(0,z_0)$, $u\colon\R^d \to \R\cup\{+\infty\}$ be a stable and regular pair potential, and $\KK_{\varsigma,\phi}$ be the Kirkwood closure process for $\varsigma=z$ and $\phi=e^{-\beta u}$. Then, for any nonnegative function $F\colon (\R^d)^n\times \Gamma\to [0,+\infty)$ and any bounded set $\Lambda\subset \R^d$ there holds 
\begin{align}\label{eq:finvolGNZ}
    \int_\Gamma 
    \sum_{x_1,\dots,x_n\in\eta_\Lambda\atop x_i\neq x_j}
    F(\xx_n;\eta_\Lambda) 
    \dK_{\varsigma,\phi}(\eta)
    =\int_{\Lambda^n} \int_\Gamma 
    F(\xx_n;\eta_\Lambda\cup\{\xx_n\})
    \frac{j_\Lambda^{(N_\Lambda(\eta)+n)}(z;\xx_n,\eta_\Lambda)}
    {\jan{N_\Lambda(\eta)}(z;\eta_\Lambda)} \dK_{\varsigma,\phi}(\eta) \dxx_n.
\end{align}
Here by abuse of notation $\eta_\Lambda$ in the argument of $\jan{N_\Lambda+n}$ (respectively $\jan{N_\Lambda}$) denotes the vector containing the points of $\eta_\Lambda$.
\end{lemma}

\begin{proof}
Let $F\colon (\R^d)^n\times \Gamma\to [0,+\infty)$, then by the defining property of the Janossy densities of the Kirkwood closure process there holds
\begin{align*}
    \int_\Gamma 
    \sum_{x_1,\dots,x_n\in\eta_\Lambda\atop x_i\neq x_j}
    F(\xx_n;\eta_\Lambda) 
    \dK_{\varsigma,\phi}(\eta)
    &=\sum_{k=0}^\infty \frac{1}{k!}
    \int_{\Lambda^k} 
    \sum_{x_1,\dots,x_n\in\yy_k \atop x_i\neq x_j}F(\xx_n;\{\yy_k\}) 
    \jan{k}(z;\yy_k) \dyy_k \\
    &=\sum_{k=n}^\infty \frac{n!}{k!} 
    \sum_{1\leq i_1<\dots < i_n\leq k}
    \int_{\Lambda^k}F(y_{i_1},\dots,y_{i_n};\{\yy_k\}) 
    \jan{k}(z;\yy_k) \dyy_k.
\end{align*}
Easy calculation gives
\begin{align*}
    &\sum_{k=n}^\infty \frac{n!}{k!}
    \sum_{1\leq i_1<\dots < i_n\leq k}
    \int_{\Lambda^k}F(y_{i_1},\dots,y_{i_n};\{\yy_k\}) 
    \jan{k}(z;\yy_k) \dyy_k \\
    &=\sum_{k=n}^\infty \frac{k(k-1)\dots (k-n+1)}{k!}
    \int_{\Lambda^k} 
    F(\yy_n;\{\yy_k\}) 
    \jan{k}(z;\yy_k) \dyy_k \\
    &=\int_{\Lambda^n}\sum_{k=0}^\infty \frac{1}{k!}
    \int_{\Lambda^k} 
    F(\xx_n;\{\yy_k\}\cup\{\xx_n\}) 
    \jan{k+n}(z;\xx_n,\yy_k) \dyy_k \dxx_n.
\end{align*}
By Remark \ref{rem:hereditary} the fraction $\jan{k+n}(\xx_n,\yy_k)/\jan{k}(\yy_k)$ is well-defined and thus
\begin{align*}
    &\int_{\Lambda^n}\sum_{k=0}^\infty \frac{1}{k!}
    \int_{\Lambda^k} 
    F(\xx_n;\{\yy_k\}\cup\{\xx_n\}) 
    \jan{k+n}(z;\xx_n,\yy_k) \dyy_k \dxx_n \\
    =\!&\int_{\Lambda^n}\sum_{k=0}^\infty \frac{1}{k!}
    \int_{\Lambda^k} 
    F(\xx_n;\{\yy_k\}\cup\{\xx_n\}) 
    \frac{\jan{k+n}(z;\xx_n,\yy_k)}{\jan{k}(z;\yy_k)}\jan{k}(z;\yy_k) \dyy_k \dxx_n.
\end{align*}
For every $\xx_n\in \Lambda^n$ define
\begin{align*}
    \tilde{F}(\eta ) :=
    F(\xx_n;\eta_\Lambda \cup\{\xx_n\}) 
    \frac{\jan{k+n}(z;\xx_n,\eta_\Lambda)}{\jan{k}(z;\eta_\Lambda)}
\end{align*}
then $\tilde{F}(\eta )=\tilde{F}(\eta_\Lambda )$, i.e. $\tilde{F}$ is local, and thus by the definition of the Janossy densities it follows that 
\begin{align*}
    &\phantom{\,=\,}\int_{\Lambda^n}\sum_{k=0}^\infty \frac{1}{k!}
    \int_{\Lambda^k} 
    F(\xx_n;\{\yy_k\}\cup\{\xx_n\}) 
    \frac{\jan{k+n}(z;\xx_n,\yy_k)}{\jan{k}(z;\yy_k)}\jan{k}(z;\yy_k) \dyy_k \dxx_n
    \\
    &=\int_{\Lambda^n} \int_\Gamma 
    F(\xx_n;\eta_\Lambda\cup\{\xx_n\})
    \frac{j_\Lambda^{(N_\Lambda(\eta)+n)}(z;\xx_n,\eta_\Lambda)}
    {\jan{N_\Lambda(\eta)}(z;\eta_\Lambda)} \dK_{\varsigma,\phi}(\eta) \dxx_n
\end{align*}
which is the right-hand side of \req{finvolGNZ}.
\end{proof}\\ \noindent
\begin{proofthm}[\ref{thm:thm3}]
Let now $\nu = \KK_{\varsigma,\phi}$ and $\al$ be the vector in $E_{\cbeta,\nu}^\infty$ defined by $\al^{(1)} (x;\eta)=e^{-\beta W(\{x \}\mid \eta)} $ and $\al^{(n)}\equiv 0$ for all $n\geq 2$.
Then for any $\Lambda$ \req{fvKNS} has a unique solution that in light of \req{jask} is given by
\begin{align}\label{eq:fvsolvartheta}
    \vartheta_\Lambda^{(n)}(z;\xx_n; \eta ) = 
    (-1)^n\frac{j_\Lambda^{(N_\Lambda(\eta)+n)}(-z;\xx_n,\eta_\Lambda)}
    {\jan{N_\Lambda(\eta)}(-z;\eta_\Lambda)}.
\end{align}
Note that $ \vartheta^{(n)}_\Lambda(z;\xx_n;\eta_\Lambda)=\vartheta^{(n)}_\Lambda(z;\xx_n;\eta)$ meaning it is a local function.
Since $\btt_\Lambda(z)$ can be written as a Neumann-series there holds
\begin{align*}
    \|\boldsymbol{\vartheta}_\Lambda\|_{\infty,\KK_{\varsigma,\phi} } 
    \leq \sum_{l=0}^\infty|z|\cbeta e^{2\beta B+1} \|z\al\|
    \leq 
    \cbeta  \frac{|z|e^{2\beta B}}{1- |z|\cbeta e^{2\beta B+1}} <+\infty.
\end{align*}
Since this bound is independent of $\Lambda$ one can choose a diagonal subsequence and find some $\tilde{\btt}\in E_{\nu}^\infty$ such that 
\begin{align*}
    \lim_{\Lambda\nearrow \R^d}
    &\sum_{n=1}\frac{1}{n!}\int_{(\R^d)^n}\int_{\Gamma}
    F^{(n)}(\xx_n;\eta ) \vartheta^{(n)}_\Lambda(z;\xx_n;\eta )\dK_{\varsigma,\phi} \dxx_n \\
    =
    &\sum_{n=1}\frac{1}{n!}\int_{(\R^d)^n}\int_{\Gamma}
    F^{(n)}(\xx_n;\eta) \tilde{\vartheta}^{(n)}(z;\xx_n;\eta)\dK_{\varsigma,\phi} \dxx_n
\end{align*}
for each $\boldsymbol{F}=(F^{(n)})_{n\geq 1}\in E_{\nu}^1$. By \req{fvKNS} one finds that 
\begin{align*}
    \sum_{n=1}\frac{1}{n!}&\int_{(\R^d)^n}\int_{\Gamma}
    F^{(n)}(\xx_n;\eta) \vartheta^{(n)}_\Lambda(z;\xx_{n};\eta)\dK_{\varsigma,\phi} \dxx_n \\
    =
    \sum_{n=1}\frac{1}{n!}&\int_{(\R^d)^n}\int_{\Gamma}
    F^{(n)}(\xx_n;\eta)
    z
    \mathds{1}_{\Lambda^n}(\xx_n)e^{-\beta W(\{x_1\} \mid \eta_\Lambda\cup \{\xx_{2,n}\})} \cdot \\
    &\sum_{l=0}^\infty\frac{1}{l!}
    \int_{(\R^d)^l}\prod_{j=1}^l f_\beta(x_1-y_j)
    \vartheta^{(n-1+l)}_\Lambda(z;\xx_{2,n},\yy_l;\eta)\yy_l\dK_{\varsigma,\phi} \dxx_n
\end{align*}
where again $\vartheta_\Lambda^{(0)}(\eta)\equiv 1$. 
Note that by \req{wcont} there holds
\begin{align*}
    \mathds{1}_{\Lambda^n}(\xx_n)
    e^{-\beta W(\{x_1\} \mid \eta_\Lambda\cup \{\xx_{2,n}\})}
    \to 
    e^{-\beta W(\{x_1\} \mid \eta\cup \{\xx_{2,n}\})}
\end{align*}
pointwise $\lebesgue^n\times \KK_{\rho,g}$-almost everywhere as $\Lambda\nearrow \R^d$, furthermore by \req{localstab2} there holds 
\begin{align*}
    \max \left\lbrace 
    \mathds{1}_{\Lambda^n}(\xx_n)e^{-\beta W(\{x_1\} \mid \eta_\Lambda\cup \{\xx_{2,n}\})},
    e^{-\beta W(\{x_1\} \mid \eta\cup \{\xx_{2,n}\})}
    \right\rbrace 
    \leq e^{2\beta B}
\end{align*} 
and thus by dominated convergence it can be concluded that
\begin{align*}
    \sum_{l=0}^\infty\frac{\cbeta^{-l}}{l!}
    \int_{(\R^d)^{l+n}}
    &\int_\Gamma\left|\left(\mathds{1}_{\Lambda^n}(\xx_n)
    e^{-\beta W(\{x_1\} \mid \eta_\Lambda \cup \{\xx_{2,n}\})} -
    e^{-\beta W(\{x_1\} \mid \eta\cup \{\xx_{2,n}\})}\right)\right.\\
    &\left.F^{(n)}(\xx_n;\eta )\prod_{j=1}^l f_\beta(x_1-y_j)\right| \dyy_l\dxx_n \dK_{\varsigma,\phi}
    \to 0
    \qquad \text{ as }\Lambda\nearrow \R^d.
\end{align*}
It follows that
\begin{align*}
    \sum_{n=1}\frac{1}{n!}&\int_{(\R^d)^n}\int_{\Gamma}
    F^{(n)}(\xx_n;\eta)
   \tilde{\vartheta}^{(n)}(z;\xx_{n};\eta)\dK_{\varsigma,\phi} \dxx_n \\
    =\sum_{n=1}\frac{1}{n!}&\int_{(\R^d)^n}\int_{\Gamma}
    F^{(n)}(\xx_n;\eta)
    z
    e^{-\beta W(\{x_1\} \mid \eta\cup \{\xx_{2,n}\})} \cdot \\
    &\sum_{l=0}^\infty\frac{1}{l!}
    \int_{(\R^d)^l}\prod_{j=1}^l f_\beta(x_1-y_j)
    \tilde{\vartheta}^{(n-1+l)}(z;\xx_{2,n},\yy_l;\eta)\yy_l\dK_{\varsigma,\phi} \dxx_n
\end{align*}
for all $\boldsymbol{F}\in E_{\nu}^1$.
The limit $\tilde{\btt}(z)$ thus satisfies \req{KNS} and since the solution of \req{KNS} is unique, one finds $\tilde{\btt}(z)=\btt(z)$. It can thus be concluded that for every $\boldsymbol{F}\in E_{\nu}^1$ there holds
\begin{equation}\label{eq:weakstarconv}
\begin{split}
    \lim_{\Lambda\nearrow \R^d}
    &\sum_{n=1}\frac{1}{n!}\int_{(\R^d)^n}\int_{\Gamma}
    F^{(n)}(\xx_n;\eta) \vartheta^{(n)}_\Lambda(z;\xx_n;\eta)\dK_{\varsigma,\phi} \dxx_n \\
    =
    &\sum_{n=1}\frac{1}{n!}\int_{(\R^d)^n}\int_{\Gamma}
    F^{(n)}(\xx_n;\eta) {\vartheta}^{(n)}(z;\xx_n;\eta)\dK_{\varsigma,\phi} \dxx_n.
\end{split}
\end{equation}
This also implies that 
\begin{align*}
    \lim_{\Lambda\nearrow\R^d}
    &\int_{(\R^d)^n}
    \int_{\Gamma}\mathds{1}_{\Lambda^n }(\xx_n)
    F(\xx_n; \eta_\Lambda )
    \vartheta_\Lambda^{(n)}(z;\xx_n; \eta ))
    \dK_{\varsigma,\phi}(\eta)\dxx_n \\
    =
    &\int_{(\R^d)^n}
    \int_{\Gamma}
    F(\xx_n; \eta ) 
    \vartheta^{(n)}(z;\xx_n;\eta) 
    \dK_{\varsigma,\phi}(\eta)\dxx_n
\end{align*}
since
\begin{align*}
    &\int_{(\R^d)^n}
    \int_{\Gamma}\mathds{1}_{\Lambda^n}(\xx_n )
    F(\xx_n; \eta)
     \vartheta_\Lambda^{(n)}(z;\xx_n; \eta))
    \dK_{\varsigma,\phi}(\eta)\dxx_n \\
    &-
     \int_{(\R^d)^n}
    \int_{\Gamma}
    F(\xx_n; \eta ) 
     \vartheta^{(n)}(z;\xx_n;\eta) 
    \dK_{\varsigma,\phi}(\eta)\dxx_n \\
    &=\int_{(\R^d)^n}
    \int_{\Gamma}
    \big(\mathds{1}_{\Lambda^n }(\xx_n)
    F(\xx_n; \eta_\Lambda  )-F(\xx_n; \eta  )\big)
     \vartheta_\Lambda^{(n)}(z;\xx_n; \eta))
    \dK_{\varsigma,\phi}(\eta)\dxx_n 
    \\&+ 
    \int_{(\R^d)^n}
    \int_{\Gamma}
    F(\xx_n; \eta ) 
    \big(\vartheta_\Lambda^{(n)}(z;\xx_n; \eta)) -\vartheta^{(n)}(z;\xx_n;\eta) \big)
    \dK_{\varsigma,\phi}(\eta)\dxx_n
\end{align*}
and the first integral goes to zero by dominated convergence since $\vartheta_\Lambda^{(n)}(z;\,\cdot\,;\,\cdot\,) $ is bounded and the second by \req{weakstarconv}. Defining 
\begin{align}\label{eq:papaKWFV}
    \kappa^{(n)}_\Lambda(z;\xx_n;\eta)
    =\frac{j_\Lambda^{(N_\Lambda(\eta)+n)}(z;\xx_n,\eta_\Lambda)}
    {\jan{N_\Lambda(\eta)}(z;\eta_\Lambda)}
    =
    (-1)^n\vartheta^{(n)}_\Lambda(-z;\xx_n;\eta)
\end{align}
and 
\begin{align}\label{eq:papaKW}
    \kappa^{(n)} (z;\xx_n;\eta)
    =(-1)^n\vartheta^{(n)}(-z;\xx_n;\eta).
\end{align}
one sees that the expressions in \req{papaKWFV} and \req{papaKW} are nonnegative and there holds
\begin{align*}
    \lim_{\Lambda\nearrow\R^d}\int_\Gamma 
   \sum_{x_1,\dots,x_n\in\eta_\Lambda \atop x_i\neq x_j}F(\xx_n;\eta_\Lambda) 
    \dK_{\varsigma,\phi}(\eta)
    &= \lim_{\Lambda\nearrow\R^d}
    \int_{(\R^d)^n}
    \int_{\Gamma}\mathds{1}_{\Lambda^n }(\xx_n )
    F(\xx_n; \eta_\Lambda)
     \kappa_\Lambda^{(n)}(z;\xx_n; \eta))
    \dK_{\varsigma,\phi}(\eta)\dxx_n\\
    &=
    \int_{(\R^d)^n}
    \int_{\Gamma}
    F(\xx_n; \eta ) 
     \kappa^{(n)}(z;\xx_n;\eta) 
    \dK_{\varsigma,\phi}(\eta)\dxx_n
\end{align*}
On the other hand one finds that for nonnegative $F$ there holds
\begin{align*}
    \lim_{\Lambda\nearrow\R^d}\int_\Gamma 
   \sum_{x_1,\dots,x_n\in\eta_\Lambda \atop x_i\neq x_j}F(\xx_n;\eta_\Lambda) 
    \dK_{\varsigma,\phi}(\eta)
    =\int_\Gamma 
    \sum_{x_1,\dots,x_n\in\eta \atop x_i\neq x_j}F(\xx_n;\eta ) 
    \dK_{\varsigma,\phi}(\eta)
\end{align*}
which proves \req{KirkwoodGNZ}.

\end{proofthm}

\section{Extension to higher order closures}
The ansatz \req{easycorrelations} with $\varphi=e^{-\beta u}$ can (in light of \req{Hpair}) be rewritten as
\begin{align}\label{eq:easycorrH}
    \rho^{(n)}(\xx_n) =  
    \varsigma^n e^{-\beta H(\xx_n)}.
\end{align}
The definition \req{easycorrH} continues to make sense when the Hamiltonian $H$ is not given by a simple pair interaction, but more complicated multi-body potentials, i.e. for each $n\geq 2$ there holds
\begin{align}\label{eq:multbodH}
    H(\xx_n) = \sum_{l=2}^n
    \sum_{1 \leq i_1<\dots <i_l \leq n}
    u^{(l)}(\xx_{i_l})
\end{align}
for some family $(u^{(l)})_{l\geq 2}$ of $l$-body interaction potentials $u^{(l)}\colon (\R^d)^l\to \R$. Here $\xx_{i_l}=(x_{i_1},\dots , x_{i_l})$. In this case for $\eta,\gamma\in\Gamma_0$ one can also define an interaction $W$ as in \req{WKprop} as
\begin{align}\label{eq:generalint}
    W(\eta\mid \gamma )=H(\eta\cup\gamma)-H(\eta)  - H(\gamma).
\end{align}
Note that $W$ can also be defined for $\gamma\in \Gamma$ under some additional conditions on $H$, e.g.~if the potentials $(u^{(l)})_{l\geq 2}$ have finite range.\\
The ansatz \req{easycorrH} with $H$ given by \req{multbodH} leads to the \emph{multi-body Kirkwood-Salsburg operator}, cf.~\cite{Moraal76}. The only difference to the two-body setting is the definition of the integral kernel of $\K$.\\
From \req{KGammadef1} and \req{KGammadef2} it follows that for a Hamiltonian $H$ given by \req{Hpair} the kernel of the Kirkwood-Salsburg equation with boundary condition $\eta \in \Gamma$ is given by
\begin{align}\label{eq:k2}
    k^{(2)}(x;\yy_k;\xx_n,\eta)=e^{-\beta W(\{x\}\mid \eta\cup\{\xx_n\})} \prod_{i=1}^k \mayer(x-y_j).
\end{align}
Using \req{mfct} one can expand the product on the right-hand side of \req{k2} to get
\begin{align*}
    k^{(2)}(x;\yy_k;\xx_n,\eta)=e^{-\beta W(\{x\}\mid \eta \cup \{\xx_n\})}
    \sum_{l=0}^k \sum_{1 \leq i_1<\dots <i_l \leq k}(-1)^{k-l}
    \prod_{j=1}^l e^{-\beta u(x -y_{i_j})}.
\end{align*}
Since the interaction $W$ is linear in the second argument there holds 
\begin{align*}
    W(\{x\}\mid \eta \cup\{\xx_n\}) +  \sum_{j=1}^lu(x -y_{i_j})
    = W(\{x\}\mid  \eta \cup\{\xx_n,\yy_{i_l}\})
\end{align*}
and thus
\begin{align}\label{eq:genkern}
    k^{(2)}(x;\yy_k;\xx_n,\eta)
    = \sum_{l=0}^k \sum_{1 \leq i_1<\dots <i_l \leq k}(-1)^{k-l}
    e^{-\beta 
    W(\{x\}\mid \eta \cup \{\xx_n,\yy_{i_l}\})}.
\end{align}
This representation of $k^{(2)}$ via \req{genkern} continues to make sense when $H$ is given by \req{multbodH} by using \req{generalint}, thus the kernel of the multi-body Kirkwood-Salsburg equations is defined as
\begin{align}\label{eq:Kirkwoodcomplicated}
    k^{(H)}(x;\yy_k;\xx_n,\eta):=
    \sum_{l=0}^k \sum_{1 \leq i_1<\dots <i_l \leq k}(-1)^{k-l}e^{-\beta 
    W(\{x\}\mid \eta \cup \{\xx_n,\yy_{i_l}\})}
\end{align}
or equivalently as
\begin{align*} 
    k^{(H)}(x;\yy_k;\xx_n,\eta):=
    \sum_{l=0}^k \sum_{1 \leq i_1<\dots <i_l \leq k}(-1)^{k-l}\exp\left(-\beta \big[ H(\{x,\xx_n,\eta,\yy_{i_l}\})
    - H(\{\xx_n,\eta,\yy_{i_l}\})\big]\right).
\end{align*}
The multi-body Kirkwood-Salsburg operator with boundary condition is then defined in an analogous way as in Subsection \ref{ss:locstab} by
\begin{align}\label{eq:KSmH1}
    (\K\bo)^{(1)}(x;\eta )=
    \sum_{k=1}^\infty\frac{1}{k!}\int_{(\R^d)^k}
    k^{(H)}(x;\yy_k;\eta)
    \theta^{(k)}(\yy_k;\eta)\dyy_k
\end{align}
and for $n\geq 1$ by
\begin{align}\label{eq:KSmH2}
     (\K\bt)^{(n+1)}(x,\xx_{n};\eta)=  
    k^{(H)}(x;\xx_n,\eta)\theta^{(n)}(\xx_{n};\eta)+
    \sum_{k=1}^\infty\frac{1}{k!}\int_{(\R^d)^k}
    k^{(H)}(x;\yy_k;\xx_n,\eta)
    \theta^{(n+k)}(\xx_{n},\yy_k;\eta)\dyy_k.
\end{align}
The case of empty boundary conditions follows by choosing $\nu = \delta_{\emptyset}$.
Now it only remains to be shown that the operator $\K$ with the kernel $k^{(H)}$ is in $\mathcal{L}(E_{\zeta,\nu})$ for some $\zeta>0$ and some measure $\nu$ on $(\Gamma,\mathscr{F})$. In this case for $|z|$ sufficiently small the solution of \req{fvKS} is again given by the Neumann-series \req{fvKSsol}. 

\begin{theorem}\label{thm:thm2}
Let $H$ be a stable and hereditary Hamiltonian given by \req{multbodH}. If there are $\zeta,\delta>0$ such that the multi-body Kirkwood-Salsburg operator $\K\colon E_\zeta\to E_\zeta$ defined by \req{KSmH1} and \req{KSmH2} is bounded with norm $\|\K\|_{E_\zeta\to E_\zeta}\leq  \delta$, then for $\varsigma < \delta$ there exists a tempered point process $\PP$ with correlation functions $(\rho^{(n)})_{n\geq 1}$ given by \req{easycorrH}.
\end{theorem}
\begin{example}\label{ex:theex}
The multi-body Kirkwood-Salsburg operator with empty boundary conditions is bounded in the following cases:
\begin{itemize}
    \item Let $H$ be given by \req{multbodH} with a family of $n$-body interactions $(u^{(n)})_{n\geq 2}$ where  
    \begin{align*}
        u^{(2)} (x,y)= u(x-y)
    \end{align*}
    for some stable and regular pair interaction $u\colon \R^d \to \R\cup\{+\infty\}$ and nonnegative translationally invariant $u^{(n)} \colon (\R^d)^n\to [0,+\infty)$ for $n\geq 3$ and in addition there is an $R>0$ with
    \begin{align*}
        u^{(n)}(\xx_n) = 0 , \qquad \text{ for all }n \geq 3
    \end{align*}
    whenever there are indices $i\neq j\in \{1,\dots,n\}$ such that $|x_i-x_j|\geq R$. Then, the multi-body Kirkwood-Salsburg operator is bounded, see \cite{Skrypnik08}. Skrypnik uses a symmetrized operator to ensure \req{locstab} holds.
    \item Let $H$ be given by \req{multbodH}  with a family of $n$-body interactions $(u^{(n)})_{n\geq 2}$ where $u^{(n)}\equiv 0$ for $n\geq 4$ and 
    \begin{align*}
        u^{(2)} (x,y)= u(x-y)
    \end{align*}
    for some stable and regular pair interaction $u\colon \R^d \to \R\cup\{+\infty\}$. Concerning $u^{(3)}$ assume further, that there is a $m \in\N$ and functions $\phi_l\colon \R^d\to\R $, $1\leq l \leq m$, such that 
    \begin{align*}
        \int_{\R^d}\left(\sum_{l=1}^m l^2 
        \phi_l^2(x)\right)^{\tfrac{1}{2}}\dx < +\infty
    \end{align*}
    and 
    \begin{align*}
        u^{(3)} (\xx_3) =   2 \sum_{l=1}^{m} \phi_l(x_2-x_1)\phi_l(x_3-x_1).
    \end{align*}
    Then, the multi-body Kirkwood-Salsburg operator is bounded, cf.~\cite{Skrypnik06}.
\end{itemize}
\end{example}
\begin{remark}
In the proof one has to first look at the locally stable case using the strategy outlined in Remark \ref{rem:kuna} before proving the general case as in Section \ref{sec:results}.
\end{remark}
\begin{remark}
When defining the operator $\K_\Gamma$ on an appropriate space $E_{\zeta,\Gamma}^\infty$ with the kernel $k^{(H)}$ one can also prove an analogous version of Theorem \ref{thm:thm3} (if the two-body potential includes a hard-core or is nonnegative), provided \req{wcont} holds, e.g.~if $u$ in the first setting of Example \ref{ex:theex} is also lower regular. For the higher-order potentials \req{wcont} trivially holds as they are of finite range. 
\end{remark}

\section*{Declarations}
Data sharing is not applicable to this article as no datasets were generated or analysed. The author states that there is no conflict of interest.


\end{document}